\documentclass[11pt]{article}
\usepackage[utf8]{inputenc}
\usepackage[english]{babel}

\usepackage{fullpage} 
\usepackage{lmodern} 
\usepackage[T1]{fontenc} 
\usepackage{amsmath}
\usepackage{amssymb}
\usepackage{microtype}
\usepackage{enumerate}
\usepackage{ctable} 
\usepackage{algorithm}
\usepackage{algorithmic}
\usepackage[algo2e,nofillcomment,ruled,linesnumbered]{algorithm2e}

\usepackage{amsthm}
\usepackage{graphicx}
\usepackage{subfig}
\DeclareCaptionType{copyrightbox}
\usepackage{color}

\newcommand{\hide}[1]{}
\newtheorem{theorem}{Theorem} 
\newtheorem{lemma}{Lemma}{\bfseries}{\itshape}
\newtheorem{observation}[lemma]{Observation} %

\usepackage{amsmath}
\usepackage{amssymb}
\usepackage{amsthm}
\usepackage{mathtools}
\usepackage{graphicx}
\usepackage{url}

\newtheorem{proposition}[lemma]{Proposition}  {\bfseries}{\itshape}

\newcommand{\red}[1]{\textcolor{red}{#1}}

\newcommand{\cnp}{\textbf{NP}}

\newcommand{\prob}{\mbox{MinConCD}}
\newcommand{\proboverlap}{\textsc{MinConCDO}}

\newcommand{\algo}{\textsc{Round}}
\newcommand{\algodp}{\textsc{DP}}

\newcommand{\black}[1]{\textcolor{black}{#1}}
\newcommand{\blue}[1]{\textcolor{blue}{#1}}
\title{Efficient Algorithms for Generating Provably Near-Optimal Cluster Descriptors for Explainability\footnote{A short version of this paper is accepted to the Thirty-Fourth AAAI Conference on Artificial Intelligence (AAAI-20)}}
\author{
\hspace{0.2in} Prathyush Sambaturu\textsuperscript{\rm 1}  
\hspace{0.2in} Aparna Gupta\textsuperscript{\rm 2}  
\hspace{0.2in} Ian Davidson\textsuperscript{\rm 3}  
\hspace{0.2in} S. S. Ravi\textsuperscript{\rm 4, \rm 5} \\ \smallskip 
\hspace{0.2in} Anil Vullikanti\textsuperscript{\rm 1, \rm 4}  
\hspace{0.2in} Andrew Warren\textsuperscript{\rm 4} \\
{\smallskip}
\textsuperscript{\rm 1}\small{Department of Computer Science, University of Virginia, Charlottesville, VA 22904} \\
\textsuperscript{\rm 2}\small{Department of Computer Science, Virginia Tech, Blacksburg, VA 24060}\\
\textsuperscript{\rm 3}\small{Department of Computer Science, UC Davis, Davis, CA 95616}\\
\textsuperscript{\rm 4}\small{Biocomplexity Institute and Initiative,
University of Virginia, Charlottesville, VA 22904}\\
\textsuperscript{\rm 5}\small{Department of Computer Science,
University at Albany--SUNY, Albany, NY 12222} {\smallskip} \\
\small{\{pks6mk, vsakumar, asw3xp\}@virginia.edu,~~
agupta12@vt.edu}, \\ \small{~~ davidson@cs.ucdavis.edu,~~ ssravi0@gmail.com}\\
}

\date{}

\begin{document}
\onecolumn

\maketitle

\begin{abstract}
Improving the explainability of the results from machine learning methods has become an important research goal. Here, we study the problem of making clusters more interpretable by extending a recent approach of [Davidson et al., NeurIPS 2018] for constructing succinct representations for clusters.  Given a set of objects $S$, a partition $\pi$ of $S$ (into clusters), and a universe $T$ of tags such that each element in $S$ is associated with a subset of tags, the goal is to find a representative set of tags for each cluster such that those sets are
pairwise-disjoint and the total
size of all the representatives is minimized. Since this problem is \cnp-hard in general, we develop  approximation algorithms with provable performance guarantees for the problem. We also show applications to explain clusters from  datasets, including clusters of genomic sequences
that represent different threat levels.
\end{abstract}

\section{Introduction}
\label{sec:intro}

As AI and machine learning (ML) methods become pervasive across all domains from health to urban planning, there is an increasing need to make the results of such methods more interpretable. Providing such explanations has now become a legal requirement in some countries \cite{goodman2016regulations}. 
Many researchers are investigating this topic 
under supervised learning,
particularly for methods in deep learning (see e.g.,
\cite{xai-17-proc,xai-18-proc}).
Clustering is a commonly used unsupervised ML technique
(see e.g., 
\cite{1742-5468-2008-10-P10008,Bolla:1604130,Fortunato201075,Tan-book-2006,Han-book-2011,Zaki-book-2014}). 
It is routinely performed on diverse kinds of datasets, sometimes after constructing network abstractions, and optimizing complex objective functions (e.g., \emph{modularity} \cite{1742-5468-2008-10-P10008}). This can often make clusters hard to interpret especially in a post-hoc analysis.
Thus, a natural question is whether it is possible to \emph{explain} a given set of clusters, using additional attributes which, crucially, 
were not used in the clustering procedure. One motivation for our work is to understand  the \emph{threat levels}
of pathogens for which genomic sequences are available.
In \cite{Jain-etal-2018a,Jain-etal-2018b,HI-2018,RW-2018}, researchers have been able to identify some genomic sequences as coming from harmful pathogens (through lab experiments and bioinformatics analysis). Understanding what makes some sequences harmful and distinguishing them from harmless sequences corresponds to the problem of interpreting the clusters.

Davidson et al. (\cite{DGR-NIPS-2018}) present the following formulation of the \emph{Cluster Description Problem} for explaining a given set of clusters. Let $S = \{s_1, \ldots,s_n\}$ be a set of $n$ objects. Let $\pi=\{C_1, \ldots, C_k\}$ be a partition of $S$ into $k \geq 2$ clusters. Let $T$ be the universe of tags such that each object $s_i\in S$ is associated with a subset $t_i\subseteq T$ of tags. 
A \textbf{descriptor} $D_i$ for a cluster $C_i$ ($1 \leq i \leq k$)
is a subset of $T$.
An object $s_j$ in cluster $C_i$ is said to be \emph{covered} by the
descriptor $D_i$ if at least one of the tags associated with $s_j$ is in $D_i$.  The goal is to find $k$ pairwise disjoint descriptors (one descriptor per cluster) so that \emph{all} the objects in $S$ are covered and the total number of tags used in all the descriptors, which will henceforth be referred to as the ``cost" of the solution, is minimized. Davidson et al. \cite{DGR-NIPS-2018} showed that even deciding whether there exists a feasible solution (with no restriction on cost) is \cnp-hard. They use Integer Linear Programming (ILP) methods to solve the problem and other relaxed versions (e.g., not requiring coverage of all objects in $S$, referred to as the ``cover-or-forget'' version, if there is no \emph{exact} feasible solution) on social  media datasets. They point out that this gives interesting and representative descriptions for clusters. However, they leave open the questions of designing efficient and rigorous approximation algorithms, which can scale to much larger datasets, and a deeper exploration of different  notions of approximate descriptions for real world datasets.

\noindent
\textbf{Our contributions.} We find in some datasets that the exact cover formulation of \cite{DGR-NIPS-2018} does not have a feasible solution. The ``cover-or-forget'' variation does give a solution, but might have highly unbalanced coverage, i.e., one cluster gets covered well, but not the others.
We extend the formulation of \cite{DGR-NIPS-2018} to address these issues, develop a suite of scalable algorithms with rigorous guarantees, and evaluate them on several real world datasets. A list of our contributions is given below.

\noindent
(1) \textbf{Formulation.} We introduce the \prob{} problem for cluster description, with \emph{simultaneous} coverage guarantees on all the clusters (defined formally in Section~\ref{prelims}). Informally, given a requirement $M_i\leq |C_i|$ for the number of objects to be covered in each cluster $C_i$, the goal is to find pairwise disjoint descriptors of minimum cost such that at least $M_i$ objects are covered in $C_i$. \prob{} gives more useful cluster descriptions in several datasets, as we discuss below. However, this problem turns out to be very hard---specifically, we show (in \cite{sambaturu2019}) that if the coverage constraints for each cluster must be met, then unless \textbf{P} = \textbf{NP}, for any $\rho \geq 1$, there is no polynomial time algorithm that can approximate the cost to within a factor of $\rho$. Therefore, we consider $(\alpha, \delta)$--approximate solutions, which ensure coverage of at least $\alpha M_i$ objects in cluster $C_i$
 ($1 \leq i \leq k$) using a cost of at most $\delta B^*$, where $B^*$
 is the minimum cost needed to satisfy the coverage requirements. 

\noindent
(2) \textbf{Rigorous algorithms.} We design a randomized algorithm, \algo{}, for \prob{}, which is based on rounding a linear programming (LP) solution. We prove that it 
gives an $(1/8, 2)$--approximation, with high probability.
For the special case of $k = 2$, we present an $(1-2/e, 1)$-approximation algorithm using techniques from submodular function maximization over a matroid (Section \ref{sec:submod}). 

\noindent
(3) \textbf{Experimental results.} We evaluate our algorithms on real and synthetic datasets (Section \ref{sec:results}). We observe that \prob{} gives more appropriate descriptions than those computed using the formulations of \cite{DGR-NIPS-2018}. 
Our results show that the approximation bounds of \algo{} are comparable with the optimum solutions of \cite{DGR-NIPS-2018} (computed using integer linear programming), and significantly better than the worst-case theoretical guarantees.  \algo{} also scales well to  instances which are over two orders of magnitude larger than those considered in \cite{DGR-NIPS-2018}. The different ``knobs'' in the formulation  (coverage  requirement per cluster, cost budget and the allowed overlap) provide a  spectrum of descriptions, which allow a practitioner to better explore and understand the clusters.
Qualitative analysis of the descriptions gives interesting insights  into the clusters. In particular, for the threat sequence dataset, our results give a small set of intuitive attributes which separate the harmful sequences from the harmless ones. 

\noindent
(4) \textbf{Extensions and additional algorithmic results.}
The theoretical guarantees of \algo{} hold only when $M_i=\Theta(|C_i|)$ for each $i$ (i.e., the requirement is to cover a constant fraction of objects in each cluster); the analysis breaks down when $M_i=o(|C_i|)$ for some clusters. 
When the $M_i$'s are arbitrary, we develop a different randomized rounding algorithm that gives an $(O(1), \eta)$--approximation, where $\eta$ is the maximum number of objects covered by any tag. 
Further, we show that \algo{} also works for a bounded overlap version of the problem for $k=2$, where cluster descriptors may overlap. 
Finally, when $|t_i|\leq\gamma$ for each $s_i$, we design a simple dynamic programming algorithm that gives an 
$O(\frac{1}{\gamma}, 1)$-approximation; see \cite{sambaturu2019}.
\smallskip

 \noindent
 \textbf{Our techniques.}
 While the cluster description problem involves \emph{covering} constraints for each cluster (as in the standard maximum coverage problem \cite{Khuller-etal-1999}), it is much harder because the
 disjointness requirement leads to independence constraints on the tags to be chosen. A standard approach for approximating covering problems is randomized rounding (see, e.g.,\cite{Srinivasan:1995:IAP:225058.225138}).
 However, here, the events that two objects $s_i$ and $s_{i'}$ are covered are dependent if $t_i\cap t_{i'}\neq\emptyset$; as a consequence, a standard Chernoff bound type concentration analysis cannot be used.
 We address these issues by observing that the events that $s_i$ and $s_{i'}$ are not covered have a specific type of dependency for which the upper tail bound by Janson and Rucinski \cite{Janson2002TheIU} gives a concentration bound. We develop a new way to analyze our randomized rounding scheme by bounding the number of objects which are not covered in each cluster.
 
 Though coverage is a submodular function (see, e.g., \cite{Calinescu:2011:MMS:2340436.2340447}), it is not clear how to use submodular function maximization techniques to simultaneously maximize the coverage within all clusters.
 We show that the ``saturation'' technique of \cite{Krause2008RobustSO}, which uses the sum of the minimum of each submodular function and a constant, can be adapted for the case $k=2$, but not for larger $k$ because the problem in \cite{Krause2008RobustSO} is for a uniform matroid constraint, whereas here we have partition matroid constraints.

\smallskip
\noindent
\textbf{Related work.}~
The topic of ``Explainable AI'' \cite{Gunning-XAI-2017} 
has attracted a lot of attention especially under supervised learning.
In particular, many researchers have studied the topic
in conjunction with methods in deep learning (e.g.,
\cite{xai-18-proc}).
To our knowledge, not much work has been done in the
context of interpreting results from clustering.
The topic of allowing a  human to interpret a given clustering
and provide suggestions for improvement
 was considered in \cite{Kuo-etal-AAAI-2017}.
Their goal was to improve the clustering quality through
human guidance, and they used constraint programming techniques
to obtain improvements.
Other methods for improving a given clustering were
considered in \cite{dang2010,qi2009}.
The notion of ``descriptive clustering" studied 
in \cite{Bich-etal-IJCAI-2018} is different from
our work; 
their idea is to allow the clustering algorithm to use both the
features of the objects to be clustered and the descriptive
information for each object.
They present methods for constructing the  Pareto frontier
based on two objectives, one based on features and the
other based on the descriptive information.
Like \cite{DGR-NIPS-2018}, the focus of our work is not on
generating a clustering; instead, the goal is to  explain the results
of clustering algorithms.
While the cluster description problem considered here uses a formulation similar to the one used in
\cite{DGR-NIPS-2018}, their focus was on expressing the problem as an ILP and solving it optimally using public domain ILP solvers.
In particular, approximation algorithms with provable performance guarantees were not considered in \cite{DGR-NIPS-2018}.
\section{Preliminaries}
\label{prelims}
\noindent
\textbf{Notation and Definitions.}~
Let $S = \{s_1, \ldots, s_n\}$ be a set of $n$ objects, and $\pi=\{C_1, \ldots, C_k\}$ be a partition of $S$ into $k \geq 2$ clusters. Let $T$ be the universe of $m$ tags such that each object $s_i\in S$ is associated with a subset $t_i\subseteq T$ of tags. A solution is a subset $X\subseteq T$, and will be represented as a partition $X=(X_1,\ldots, X_k)$, where $X_{\ell}$ is the descriptor (i.e., subset of tags) used for cluster  $C_{\ell}$. 
We say that $s_i \in S$ is \textit{covered} by a set $X \subseteq T$ of tags if $X \cap t_i \neq \emptyset$. Let $E(j) = \{s_i : j \in t_i\}$ be the set of all objects that can be covered by the tag $j \in T$. Let $\eta = \max_{j} |E(j)|$ denote the maximum number of objects covered by any tag in $T$.
Let $\gamma= \max_i |t_i|$ denote the maximum number of tags associated with any object in $S$.
Objects $s_i, s_{i^{'}} \in S$ are said to be \textit{dependent} if $t_i \cap t_{i^{'}} \neq \emptyset$, i.e., if their tag sets overlap. 
Let $\Delta(i) = |\{i^{'}: t_i \cap t_{i^{'}} \neq \emptyset\}|$ denote the degree of dependence of $s_i$, and let $\Delta = \max_i \Delta(i)$ be the maximum dependence.
Finally, for a solution $X=(X_1,\ldots,X_k)$, let $V_{\ell}(X)=\{s_i \in C_{\ell}: t_i \cap X_{\ell} \neq \phi\}$ be the subset of objects in $C_{\ell}$ covered by $X$, $1 \leq \ell \leq k$. 
For any integer $k \geq 1$, we use $[k]$ to denote the set $\{1,\ldots,k\}$.

\smallskip

\noindent
\textbf{Problem statement.}~
Our objective is to find a solution $X$ that \emph{simultaneously} ensures high coverage $|V_{\ell}(X)|$ in each cluster $C_{\ell}$, 
$1 \leq \ell \leq k$. An obvious choice is to consider a max-min  type of objective $X=\mbox{argmax}\min_{\ell}|V_{\ell}(X)|$ (see, e.g., \cite{DBLP:journals/corr/abs-1711-06428}). However, this doesn't allow  domain specific coverage requirements (e.g., higher coverage for the cluster of threat sequences in genomic data). Therefore, we consider a more general formulation, which specifies a coverage requirement for each cluster.

\smallskip
\noindent
\textbf{Minimum Constrained Cluster Description (\prob{})}

\noindent
\underline{Instance}: A set $S = \{s_1, \ldots, s_n\}$ of objects, a  partition $\pi=\{C_1, \ldots, C_k\}$ into $k \geq 2$ clusters, a universe $T$ of $m$ tags, tag set $t_i \subseteq T$ for each object $s_i$ and parameter $M_{\ell}$ for each cluster $C_{\ell}$, $1 \leq \ell \leq k$.

\noindent
\underline{Requirement}: Find a solution $X=(X_1,\ldots,X_k)$ that
minimizes the cost $\sum_{\ell=1}^{k}|X_{\ell}|$ and satisfies 
the following constraints:
(i) the subsets in $X$ are pairwise-disjoint and
(ii) for each cluster $C_{\ell}$,
$|V_{\ell}(X)| \geq M_{\ell}$.

\smallskip
\noindent
\textbf{Comparison with the formulation in \cite{DGR-NIPS-2018}.}~
The main problem considered in \cite{DGR-NIPS-2018} is to
minimize the cost  (i.e., $\sum_{\ell}|X_{\ell}|$)
under the constraint that all the objects in $S$ are covered.
This was formulated as an integer linear program (ILP).
That reference also presented an ILP  for minimizing
the cost under the requirement that at least a total of
$\alpha |S|$ objects are covered  over all the clusters 
for a given $\alpha$, $0 < \alpha \leq 1$.
(This was called the ``cover or forget" formulation in \cite{DGR-NIPS-2018}.)
One difficulty with this formulation is that solutions that
satisfy the total coverage
requirement may cover a large percentage of the objects
in some clusters while covering only a small percentage of those
in other clusters.
(We will present an example of this phenomenon using a real life dataset
in Section~\ref{sec:results}.)
Our formulation avoids this difficulty by allowing the specification of the coverage requirement for each cluster separately.
As mentioned in Section~\ref{sec:intro}, our formulation in conjunction with the disjointness requirement introduces additional challenges in developing efficient approximation algorithms with provable performance guarantees.

\smallskip
\noindent
\textbf{Approximation algorithms.} 
Here, we present a simple result regarding the hardness
of approximating \prob.

\begin{proposition}\label{pro:approx_hardness}
If the coverage constraints specified in \prob{} must be met exactly,
then  for any $\rho \geq 1$, it is \cnp-hard to obtain a $\rho$--approximation for the cost of the solution (i.e., the total number of tags used in the solution).
This result holds even when $k$ = 2 and
each object has at most three tags.
\end{proposition}
\begin{proof}
Using a reduction from 3SAT, Theorem~3.1  in \cite{DGR-NIPS-2018} 
shows that it is \cnp-hard to determine whether there are  disjoint descriptors for two clusters such that all the objects in the
two clusters are covered.
They denote this problem by DTDF.
The instances of DTDF produced by the reduction from 3SAT 
satisfy two properties: (i) there is no restriction on
 the number of tags that can be used by the descriptors
 and (ii) each object in the DTDF instance has at most 3 tags.
 We use the same reduction and interpret the resulting DTDF instance
 as an instance of \prob{} as follows: set the coverage requirement 
 for each cluster $C_{\ell}$ to $|C_{\ell}|$, for $\ell = 1,\,2$,
 It is easy to see that there is a solution to the 3SAT instance iff an approximation  algorithm with any performance guarantee $\rho \geq 1$ produces a solution for the \prob{} instance.
In other words, the existence of a polynomial time $\rho$--approximation algorithm for \prob{} would imply a polynomial time algorithm for 3SAT.
Since the latter problem is \cnp-hard, the proposition follows.
\end{proof}

Therefore, if the coverage requirements
must be met, then the cost cannot be approximated to within any factor
$\rho \geq 1$, unless \textbf{P} = \cnp. 
Therefore,
we study bi-criteria approximation algorithms. We say that a solution $X$ is an $(\alpha, \delta)$-approximation if 
(i) for each cluster $C_{\ell}$, 
$|V_{\ell}(X)| \geq \alpha M_{\ell}$ and 
(ii) $\sum_{\ell} |X_{\ell}| \leq \delta B^*$, 
where $B^*$ is the optimal cost.

\smallskip
\noindent
\textbf{Other variations.}
Another variation we will explore, referred to as \proboverlap, allows limited overlap between different descriptors. Given input parameters $M_{\ell}$ for $\ell\in[k]$ and overlap limit $B_o$, the objective here is to find a solution $X=(X_1,\ldots,X_k)$ of minimum cost such that 
$|V_{\ell}(X)|\geq M_{\ell}$ for each $\ell$ and 
$\sum_{\ell\neq\ell'} |X_{\ell}\cap X_{\ell'}|\leq B_o$.

\section{Algorithms}
\subsection{Algorithm \algo{}: approximation using Linear Programming and Rounding}
\label{methods1}

Our approach for approximating \prob{} is based on LP relaxation and then rounding the fractional solution. This is a common approach for many combinatorial optimization problems, especially those with covering constraints (see, e.g., \cite{DBLP:books/daglib/0030297}). However, the disjointness requirement for descriptors poses a challenge in terms of dependencies and requires a new method of rounding. We start with an ILP formulation. 
\smallskip

\noindent
\textbf{ILP Formulation.} 
For each $j\in T$ and $\ell\in[k]$, $x_{\ell}(j)$ is an indicator
(i.e., \{0,1\}-valued) variable, which is $1$ if tag $j\in X_{\ell}$.
We have an indicator variable $z(i)$ for each $s_i \in S$, which is $1$ if object $s_i$ is covered. 
The objective and constraints of the ILP formulation are as follows.
\begin{center}
$(\mathcal{IP})$ \hspace*{0.2in} Minimize   
$\displaystyle{\sum_{\ell=1}^{k}\sum_{j \in T} x_{\ell}(j)}$~~
such that \\ \smallskip
\begin{tabular}{l}
$\displaystyle{\forall \ell,\ \forall s_i\in C_{\ell}: \sum_{j\in t_i} x_{\ell}(j) ~\geq~ z(i)}$\\
$\displaystyle{\forall \ell: \sum_{s_i\in C_{\ell}} z(i) ~\geq~ M_{\ell}}$\\ [2ex]
$\displaystyle{\forall j: \sum_{\ell} x_{\ell}(j) ~\leq~  1}$,~~ All variables $\in\{0, 1\}$ \\
\end{tabular}
\end{center}


\begin{algorithm2e}[tb]
\DontPrintSemicolon
\SetKwInOut{Input}{Input}\SetKwInOut{Output}{Output}
\Input{
$S$, $\pi=\{C_1,\ldots, C_k\}$, $T$, $M_{\ell}$ for each $\ell=1,\ldots,k$.
(Note:~ $|S| = n$.)
}
\Output{
$X=(X_1,\ldots,X_k)$
}
Let $\mathcal{P}$ be a linear relaxation of the ILP $\mathcal{IP}$, obtained by requiring all variables to be in $[0, 1]$ (instead of being binary).\;
Solve $\mathcal{P}$. If  it is not feasible,
return ``no feasible solution''. Else, let $x^*, z^*$ denote the optimal fractional solution and $B$ denote the associated cost.\;
For all $j$, and for all $\ell$, set $x_{\ell}(j) = x^*_{\ell}(j)/2$, and for all $s_i$ set $z(i) = z^*(i)/2$.\;
\For{$4\ln{n}$ ~\mbox{times}}{
\For{$j \in T$ ~\mbox{and}~ $\ell=1,\ldots,k$}{
With probability $x_{\ell}(j)$, round $X_{\ell}(j)=1$ and $X_{\ell'}(j)=0$ for all $\ell'\neq\ell$.\;
With probability $1-\sum_{\ell} x_{\ell}(j)$, set $X_{\ell'}(j)=0$ for all $\ell'$.\;
}
\For{$s_i\in S$}{
If $X_{\ell}(j)=1$ for some $j\in t_i$, define $Z(i)=1$.\;
}
For each $\ell$, define 
$Z_{\ell} = \sum_{s_i\in C_{\ell}} Z(i)$.\;
If $Z_{\ell}\geq M_{\ell}/8$ for each $\ell$, and $\sum_{\ell}\sum_j X_{\ell}(j)\leq 2B$, return $X$ as the solution and \textbf{stop}.\;
}
Return failure.\;
\caption{Algorithm \algo{}}
\label{alg:round}
\end{algorithm2e}

Algorithm \ref{alg:round} describes the steps of \algo{}. The linear program $\mathcal{P}$ (from Step~1) can be solved (Step 2) using standard techniques in polynomial time (e.g., \cite{KT-2006}) and
a fractional solution to the variables of $\mathcal{P}$ can be obtained efficiently whenever there is a feasible solution. We analyze the performance of \algo{} in
Theorem \ref{theorem:round1}.
Most of our discussion will focus on analyzing the solution $X=(X_1,\ldots,X_k)$ computed in any iteration of Steps 4--14 of Algorithm~\ref{alg:round}.
For each $\ell$, define $Z_{\ell} = \sum_{s_i\in C_{\ell}} Z(i)$.
A proof of the following lemma appears in \cite{sambaturu2019}.

\begin{lemma}
\label{lemma:Z_kprob}
For each $\ell\in[k]$, the expected number of objects covered in cluster $C_{\ell}$ by a solution $X$ in any round of Step 4 of algorithm \algo{} is at least $M_{\ell}/{4}$.
\end{lemma}
\begin{proof}
The probability that $s_i \in C_{\ell}$ is covered after rounding is
\begin{eqnarray*}
    Pr(Z(i)=1) &=& 1 - \prod_{j \in t_i} (1-x_{\ell}(j)) \\
                 &\geq& 1 - e^{-\sum_{j \in t_i} x_{\ell}(j)} \\
                 &\geq&  1 - e^{-z(i)} \geq   \frac{z^*(i)}{2}  
\end{eqnarray*}
The first inequality above follows from the fact that $1-x\leq e^{-x}$ for $x\in[0, 1]$, whereas the second inequality follows from the fact that $e^{-x}\leq 1-x/2$ for $x\in[0, 1/2]$. The last inequality follows from the scaling step.
Therefore, the expected number of objects covered in $C_{\ell}$ is
\begin{equation*}
    E[Z_{\ell}] = \sum_{s_i \in C_{\ell}} E[Z(i)] \geq  \sum_{s_i \in C_{\ell}} \frac{z^*(i)}{2} \geq \frac{M_{\ell}}{4},
\end{equation*}
since the LP ensures $\sum_{s_i\in C_{\ell}} z^*(i)\geq M_{\ell}$.
\end{proof}

\noindent
\textbf{Challenge in deriving a lower bound on the number of objects covered in each cluster.}
Lemma \ref{lemma:Z_kprob} implies that, in expectation, a constant fraction of the objects in each cluster are covered. If the variables $Z(i)$ were all independent, we could use a Chernoff bound (see, e.g., \cite{DBLP:books/daglib/0025902}) to show that $Z_{\ell}$ is concentrated around $E[Z_{\ell}]=\Theta(M_{\ell})$. However, $Z(i)$ and $Z(i')$ are dependent when $t_i\cap t_{i'}\neq\emptyset$. Therefore, the standard Chernoff bound  cannot be used in this case, and a new approach is needed. We address this issue by considering the number of objects which are \emph{not covered} in $C_{\ell}$. There are dependencies in this case as well; however, they are of a special type, which can be handled by Janson's upper tail bound \cite{Janson2002TheIU}, which is one of the few known concentration bounds for dependent events.

For $s_i \in C_{\ell}$, let $Y(i) = 1-Z(i)$ be an indicator for $s_i$ not covered by the solution $X$. 
We apply the upper tail bound of \cite{Janson2002TheIU} to obtain a concentration bound on $Y_{\ell} = \sum_{s_i\in C_{\ell}} Y(i)$. We first observe that the dependencies among the $Y_{\ell}$ variables are of the form considered in \cite{Janson2002TheIU}.
Let $\Gamma=T$, and let $\xi_{j\ell}$ be an indicator that is $1$ if $X_{\ell}(j)=0$. For a fixed $\ell$, the variables $X_{\ell}(j)$ are independent over all $j$, since \algo{} rounds the variables for each $j$ independently. Hence, for a fixed $\ell$, the random variables $\xi_{j\ell}$ are all independent.
Then, for $s_i\in C_{\ell}$, $Y(i)=\prod_{j\in t_i} \xi_{j\ell}$.
This implies that $Y(i)$ and $Y(i')$ are independent if $t_i\cap t_i'=\emptyset$.
Therefore, the random variables $Y(i)$ are of the type considered in \cite{Janson2002TheIU}.
Let $\Delta$ be the maximum number of sets $t_{i'}$ which intersect with any $t_i$,
as defined earlier in Section \ref{prelims}, and let $\lambda=E[Y_{\ell}]$.
Then, the bound from \cite{Janson2002TheIU} gives
\[
\Pr[Y_{\ell} \geq \lambda + t] \leq (\Delta + 1)\,\mbox{exp}\Big(-\frac{t^2}{4(\Delta+1)(\lambda+t/3)}\Big).
\]
We use this bound in the lemma below.

\begin{lemma}
\label{Y1bound}
Assume $M_{\ell} \geq a|C_{\ell}|$ for all $\ell\in[k]$ for a constant $a\in(0,1]$, and let $(\Delta+1) \leq \min_{\ell\in[k]}\frac{d |C_{\ell}|}{\log n}$, where $d \leq \frac{a^2}{576}$. 
Let $t = M_{\ell}/{8}$. Then, for any fixed $\ell\in[k]$ and  any round of Step 4 of \algo,
$\Pr(Y_{\ell} \geq E[Y_{\ell}] +t) \leq \frac{1}{n}$.
\end{lemma}
\begin{proof}
We consider a fixed $\ell$ here.
We have $\Pr[Y(i)=1] = 1-\Pr[Z(i)=1] \leq 1-z^*(i)/2$, from the proof of Lemma \ref{lemma:Z_kprob}.
Next, $\sum_{s_i \in C_{\ell}} Y(i) + Z(i) = Y_{\ell} + Z_{\ell} = |C_{\ell}|$. Therefore, from Lemma \ref{lemma:Z_kprob}
\begin{equation*}
     E[Y_{\ell}] = |C_{\ell}| - E[Z_{\ell}]  \leq |C_{\ell}| - \frac{M_{\ell}}{4} \leq (1-\frac{a}{4}) |C_{\ell}|,
\end{equation*}
since $M_{\ell} \geq a|C_{\ell}|$.
Let $\lambda = E[Y_{\ell}]$. Then
\[
	(\Delta+1) \lambda ~\leq~ \frac{d |C_{\ell}|^2 (1-\frac{a}{4})}{\log n} ~\leq~ \frac{d |C_{\ell}|^2 (1-\frac{a}{4}) a^2}{a^2\log n}
\]
\[
	~\leq~ \frac{d(1-\frac{a}{4}) 64 t^2}{a^2 \log n}
\]
Also
\begin{equation*}
(\Delta+1)\frac{t}{3} \leq  \frac{d |C_1| t a}{3 a \log n}  \leq \frac{d M_1 t}{3a \log n} \\ \nonumber
		    = \frac{8d t^2}{3a\log n} \\ \nonumber
\end{equation*}
Therefore,
\begin{equation}
	4(\Delta+1) (\lambda+\frac{t}{3})  \leq \frac{8 d t^2}{a \log n}\Big[\frac{8(1-a/4)}{a} + \frac{1}{3}\Big] \\ \nonumber
\end{equation}
Putting these together, we have
\begin{equation*}
	\frac{t^2}{4(\Delta+1) (\lambda+\frac{t}{3})}  \geq \frac{\log n}{\frac{8d}{a} [\frac{8(1-a/4)}{a} + \frac{1}{3}]} \geq 2 \log n,
\end{equation*}
where the last inequality follows because $d \leq \frac{a^2}{576}$; thus,  $\frac{8d}{a} \Big[\frac{8(1-a/4)}{a} + \frac{1}{3}\Big]\leq \frac{1}{2}$.
Applying Janson's upper tail bound,
\begin{eqnarray*}
	 \Pr(Y_{\ell} \geq \lambda +t) & \leq& (\Delta+1) \mbox{exp}\Big(\frac{-t^2}{4(\Delta+1)(\lambda+t/3)}\Big) \\ \nonumber
	 						 & \leq& (\Delta+1)\, \mbox{exp}(-2 \log n) 
	 						 \leq \frac{(\Delta+1)}{n^2}\\ \nonumber
	 						 &\leq& \frac{M_{\ell}}{n^2} 
	 						  ~\leq~ \frac{1}{n},
\end{eqnarray*}
where the last inequality is because $M_{\ell}$ $\leq$ $|C_{\ell}|\leq n$.
\end{proof}

\begin{theorem}
\label{theorem:round1}
Suppose an instance of \prob{} satisfies the following conditions: (1) $M_{\ell} \geq a |C_{\ell}|$ for all $\ell\in[k]$, and for some constant $a\in(0, 1]$, and (2) $(\Delta+1) \leq \min_{\ell}\frac{d |C_{\ell}|}{\log n}$ and $d \leq \frac{a^2}{576}$, and (3) $k\leq n/4$.
If the LP relaxation $(\mathcal{P})$ is feasible, then
with probability at least $1-\frac{1}{n}$, algorithm \algo{} successfully returns a solution $X$, which is an $(1/8, 2)$--approximation.
\end{theorem}
\begin{proof}
We analyze the properties of a solution $X$ computed in each round 
of Step 4 of \algo{}.
From Lemma \ref{Y1bound}, for any $\ell\in[k]$, $Y_{\ell} \leq E[Y_{\ell}] + \frac{M_{\ell}}{8}$ with probability at least $1-\frac{1}{n}$. Substituting $E[Y_{\ell}] \leq |C_{\ell}| - \frac{M_{\ell}}{4}$ (shown in the proof of Lemma \ref{Y1bound}) in the above equation we have,
\begin{equation}
	Y_{\ell} \leq |C_{\ell}| - \frac{M_{\ell}}{4} + \frac{M_{\ell}}{8} \leq |C_{\ell}|-\frac{M_{\ell}}{8} \nonumber
\end{equation}
with the same probability, for each $\ell$. Therefore, 
\begin{equation*}
  Z_{\ell} \geq |C_{\ell}| - Y_{\ell} \geq |C_{\ell}| - \Big(|C_{\ell}| - \frac{M_{\ell}}{8}\Big)  \geq \frac{M_{\ell}}{8},
\end{equation*}
for each $\ell$, with probability at least $1-1/n$.
This implies $\Pr[Z_{\ell} < \frac{M_{\ell}}{8}]\leq 1/n$, so that $\Pr[\exists\ \ell\mbox{ with }Z_{\ell} < \frac{M_{\ell}}{8}] \leq k/n$. Therefore, with probability at least $1-k/n$, for all $\ell\in[k]$, we have $Z_{\ell}\geq \frac{M_{\ell}}{8}$.

Next, we consider the cost of the solution
(i.e., the total number of tags used). The rounding ensures that $\Pr[X_{\ell}(j)=1]=x_{\ell}(j)$, for each $\ell, j$. Thus, by linearity of expectation, the expected cost of the solution is
\begin{equation*}
E\Big[\sum_{\ell}\sum_j  X_{\ell}{(j)}\Big] 
= \sum_j \sum_{\ell}x_{\ell}{(j)} \leq B \nonumber
    \end{equation*}
By Markov's inequality,  $Pr[\sum_{\ell}\sum_j  X_{\ell}{(j)} > 2B] \leq \frac{1}{2}$. 

Putting everything together, for each round, the probability of success (i.e., the cost is at most $2B$ and $Z_{\ell}\geq M_{\ell}/8$ for each $\ell$) is at least $\frac{1}{2} - \frac{k}{n}\geq \frac{1}{4}$, since $k\leq n/4$. Therefore, the probability that at least one of the $4\ln{n}$ rounds is a success is at least $1-(\frac{3}{4})^{4\ln{n}}\geq 1-\frac{1}{n}$.
\end{proof}

\subsection{Rounding algorithm when $M_{\ell}$ values are arbitrary}
\label{sse:rounding_two}

Algorithm \algo{} is won't ensure coverage of $\Omega(M_{\ell})$ within each cluster $C_{\ell}$, unless $\Delta$ (the degree of dependence between the sets of tags) becomes very small, which can limit the utility of the approach. We now present a different rounding method for this case, which allows $\Delta$ to be as in Theorem \ref{theorem:round1}, but leads to a worse approximation ratio. Our algorithm involves the following steps.

\begin{enumerate} 
\item
If the LP $(\mathcal{P})$ is not feasible,
return ``no feasible solution''. Else, let $x^*, z^*$ denote the optimal fractional solution for LP$(\mathcal{P})$.
\item
Run the following steps $4\ln{n}$ times.
\begin{itemize}
\item
For each $s_i \in S$, set $Z(i)=1$ with probability $z^*(i)$. Let $C'_{\ell}=\{s_i\in C_{\ell}: Z(i)=1\}$ for each $\ell$.
\item
For each $\ell, j$, define $x'_{\ell}(j)=0$ if $E(j)\cap \Big(\cup_{\ell} C'_{\ell}\Big)=\emptyset$; else define
\[
x'_{\ell}(j) = \frac{x^*_{\ell}(j)}{2}\cdot \max\Big\{\frac{1}{z^*(i)}: s_i\in C'_{\ell}, j\in t_i\Big\}
\]
\item
For every $j\in T$, and for each $\ell=1,\ldots,k$: with probability $x'_{\ell}(j)$,  round $X_{\ell}(j)=1$ and $X_{\ell'}(j)=0$ for all $\ell'\neq\ell$. With probability $1-\sum_{\ell} x'_{\ell}(j)$, set $X_{\ell'}(j)=0$ for all $\ell'$.

\item
For each $s_i$, define $Z'(i)=1$ if $X_{\ell}(j)=1$ for some $j\in t_i$, and $Z'_{\ell} = \sum_{s_i\in C_{\ell}} Z'(i)$, for each $\ell$.\smallskip
\item
If $Z'_{\ell}\geq c\cdot M_{\ell}$ for each $\ell$ and $\sum_{\ell}\sum_j X_{\ell}(j)\leq \eta B$, return $X$ as the solution and \textbf{stop}.
\end{itemize}
\item
Return failure.
\end{enumerate}


\begin{lemma}
\label{lemma:xprime}
For each $\ell\in[k]$ and $j\in T$,  $x'_{\ell}(j)\leq\frac{1}{2}$
\end{lemma}
\begin{proof}
For any $j$, we can assume, wlog, that $\sum_{j\in t_i} x_{\ell}(j)=z(i)$ for all $i$
(else, if $\sum_{j\in t_i} x_1(j) > z(i)$, we can increase $z(i)$ and still keep the
solution feasible).
This means $x_{\ell}(j)\leq z(i)$ for all $i\in E(j)$.
If $E(j)\cap \big(\cup_{\ell} C'_{\ell}\big)=\emptyset$, $x'_{\ell}(j)=0\leq\frac{1}{2}$. So assume
$E(j)\cap \big(\cup_{\ell} C'_{\ell}\big)\neq\emptyset$. 
Therefore, for each $s_i\in C'_{\ell}$ such that $j\in t_i$, we have 
$\frac{x_{\ell}(j)}{z(i)}\leq 1$, which implies 
\[
x_{\ell}(j)\max\{\frac{1}{z(i)}: s_i\in C'_1, j\in t_i\} \leq 1.
\]
This implies $x'_{\ell}(j)\leq\frac{1}{2}$.
\end{proof}

\begin{lemma}
\label{lemma:sum_i}
For each $\ell\in[k]$ and $s_i\in C'_{\ell}$: $\sum_{j\in t_i} x'_{\ell}(j) \geq \frac{1}{2}$.
\end{lemma}
\begin{proof}
Consider any $s_i\in C'_{\ell}$. 
Then, $x'_{\ell}(j)\geq \frac{x_{\ell}(j)}{2z(i)}$, which implies
\[
\sum_{j\in t_i} x'_{\ell}(j)\geq \frac{1}{2}\sum_{j\in t_i} \frac{x_{\ell}(j)}{z(i)}\geq \frac{1}{2}
\]
\end{proof}

\begin{lemma}
\label{lemma:expxprime}
For all $\ell, j$: $E[x'_{\ell}(j)]\leq \eta\cdot \frac{x_{\ell}(j)}{2}$.
\end{lemma}
\begin{proof}
By construction, we have 
$x'_{\ell}(j) = \frac{x_{\ell}(j)}{2}\cdot \max\{\frac{1}{z(i)}: s_i\in C'_{\ell}, j\in t_i\}$.
Therefore,
\[
x'_{\ell}(j) \leq \frac{x_{\ell}(j)}{2}\sum_{s_i\in C'_{\ell}, j\in t_i} \frac{1}{z(i)}
\]
This implies
\begin{eqnarray*}
E[x'_{\ell}(j)] &\leq& \frac{x_{\ell}(j)}{2} \sum_{i\in C_1, j\in t_i} \frac{1}{z(i)}\Pr[s_i\in C_1] \\
&=& \frac{x_1(j)}{2} \sum_{s_i\in C_1, j\in t_i}  \frac{1}{z(i)}z(i)\\
&=&\frac{x_{\ell}(j)}{2} \sum_{s_i\in C_1, j\in t_i} 1\\
&\leq& \frac{x_1(j)}{2} |E(j)|
\;\leq\; \eta\frac{x_{\ell}(j)}{2}
\end{eqnarray*}
\end{proof}

\begin{observation}
\label{obs:x1x2}
For any fixed $\ell$, the variables $X_{\ell}(j)$ are all independent.
Further, for each $j$, $\sum_{\ell}X_{\ell}(j)\leq 1$.
\end{observation}

\begin{lemma}
\label{lemma:M1prime}
Assume $M_{\ell}\geq \frac{2}{\epsilon^2}\log{n}$ for all $\ell\in [k]$.
Then,
$M'_{\ell}\geq (1-\epsilon)M_{\ell}$ with probability at least $1-\frac{1}{n}$.
\end{lemma}
\begin{proof}
Consider any fixed $\ell$. We have
\begin{equation*}
M_{\ell}^{'} = |C_{\ell}^{'}| = \sum_{s_i \in C_{\ell}} Z_i 
\end{equation*}
The expected value of $M_{\ell}^{'}$ is
\begin{eqnarray*}
	E[M_{\ell}^{'}] &=& \sum_{s_i \in C_{\ell}} E[Z_i] \\ 
	           &=& \sum_{s_i \in C_{\ell}} z_i
	            \geq M_{\ell} \;  (from \; constraints)
\end{eqnarray*}
By the Chernoff bound \cite{DBLP:books/daglib/0025902},
\begin{eqnarray*}
    Pr(M_{\ell}^{'} < (1-\epsilon) E[M_{\ell}^{'}]) & \leq& \mbox{exp}(- \frac{E[M_{\ell}^{'}] \epsilon^2}{2}) \\ \nonumber
                    & \leq& \mbox{exp}(-\frac{M_{\ell}\epsilon^2}{2}) \\
                      & \leq& \mbox{exp}(-\log n) 
                      \leq \frac{1}{n} 
\end{eqnarray*}
Therefore, $Pr(M_{\ell}^{'} \geq (1-\epsilon) M_{\ell}) \geq 1 - \frac{1}{n}$.
\end{proof}

For notational simplicity, we discuss the analysis for $k=2$ in rest of this section.
Let $Y_1(i)=1$ if $X_1(j)=0$ for all $j\in t_i$. Let $Y_1 = \sum_{s_i\in C'_1} Y_1(i)=|V_1(X_1,  X_2)|$
denote the  number of sets in $C_1$ which are covered by $X_1$. Similarly, we define $Y_2(i)$ and $Y_2$. 

\begin{lemma}
\label{lemma:Y1-bound}
Suppose $\Delta+1\leq 0.01M'_1/\log{n}$. Then,
$\Pr[Y_1 > 0.86 M'_1] \leq \frac{1}{n}$ and
$\Pr[Y_2 > 0.86 M'_2] \leq \frac{1}{n}$.
\end{lemma}
\begin{proof}
First, observe that 
\[
\Pr[Y_1(i)=1] = \prod_{j\in t_i}(1-x'_1(j)) \leq \mbox{exp}(-\sum_{j\in t_i} x'_1(j)) \leq e^{-1/2},
\]
since $\sum_{j\in t_i} x'_1(j)\geq 1/2$, by Lemma \ref{lemma:sum_i}.
Therefore, $\lambda = E[Y_1]  = \sum_{s_i\in C'_1} E[Y_1(i)] \leq \frac{1}{\sqrt{e}} M'_1$. 

We apply the upper tail bound of \cite{Janson2002TheIU} for $Y_1$.
Let $\Gamma=T$, and let $\xi_j$ be an indicator that is $1$ if $X_1(j)=0$.
Then, $Y_1(i)=\prod_{j\in t_i} \xi_j$ (by independence of $X_1(j)$'s, and therefore, of $\xi_j$'s).
Further, $Y_1(i)$ and $Y_1(i')$ are independent if $t_i\cap t_i'=\emptyset$.
Therefore, the random variables $Y_1(i)$ are of the type considered in \cite{Janson2002TheIU}.
Let $\Delta$ be as defined earlier.
Then, the bound from \cite{Janson2002TheIU} gives
\[
\Pr[Y_1 \geq \lambda + t] \leq (\Delta + 1)\mbox{exp}\Big(-\frac{t^2}{4(\Delta+1)(\lambda+t/3)}\Big)
\]

We choose $t=cM'_1$, with $c=0.25$. By assumption, $\Delta+1\leq \frac{dM'_1}{\log{n}}$ for $d=0.01$.
Then, we have
\[
(\Delta+1)\lambda \leq \frac{dM'_1}{\log{n}}\frac{M'_1}{\sqrt{e}} = \frac{d}{c^2\sqrt{e}\log{n}} (cM'_1)^2
= \frac{d}{c^2\sqrt{e}\log{n}} t^2
\]
and
\[
(\Delta+1)t \leq \frac{dM'_1}{\log{n}}t = \frac{d}{c\log{n}}cM'_1t = \frac{d}{c\log{n}}t^2
\]
Therefore,
\begin{eqnarray*}
4(\Delta+1)(\lambda + t/3) &\leq& 4\Big(\frac{d}{c^2\sqrt{e}\log{n}} + \frac{d}{3c\log{n}}\Big)t^2 \\
&=& \frac{4d}{c\log{n}}\Big(\frac{1}{c\sqrt{e}}+\frac{1}{3}\Big)t^2,
\end{eqnarray*}
which implies
\[
\frac{t^2}{4(\Delta+1)(\lambda + t/3)} \geq \frac{\log{n}}{\frac{4d}{c}\Big(\frac{1}{c\sqrt{e}}+\frac{1}{3}\Big)}
\geq 2\log{n},
\]
for $c=0.25, d=0.01$.

Since $\lambda\leq \frac{1}{\sqrt{e}}M'_1\leq 0.607 M'_1$, and $t=0.25 M'_1$,
this, in turn, implies
\begin{eqnarray*}
\Pr[Y_1 > 0.86 M'_1] &\leq& \Pr[Y_1 \geq 0.857 M'_1] \leq (\Delta+1)e^{-2\log{n}} \\
&\leq& M'_1\frac{1}{n^2}\leq \frac{1}{n}
\end{eqnarray*}
\end{proof}

\noindent
\begin{theorem}\label{theorem:round2}
Suppose the LP relaxation $(\mathcal{P})$ is feasible for $B, M_1, M_2$,
and $M_1, M_2\geq \frac{2}{\epsilon^2}\log{n}$.
With probability at least $\frac{1}{2}-\frac{2}{n}$, the above rounding algorithm gives a solution $X_1, X_2$
such that $|V_1|\geq 0.14(1-\epsilon)M_1$, $|V_2|\geq 0.14(1-\epsilon)M_2$ and
$|X_1|+|X_2|\leq \eta B$.
\end{theorem}
\begin{proof}
By Lemma \ref{lemma:Y1-bound}, we have $Y_1\leq 0.86 M'_1$ and $Y_2\leq 0.86 M'_2$ with probability at least $1-\frac{2}{n}$.
We have $|V_1\cap C'_1| + Y_1= M'_1$, which implies $|V_1| \geq 0.14 M'_1$ and $|V_2| \geq 0.14 M'_2$
with this probability.

By Lemma  \ref{lemma:M1prime}, $M'_1\geq(1-\epsilon)M_1$ and $M'_2\geq(1-\epsilon)M_2$
with probability $1-\frac{2}{n}$.
Therefore, with probability at least $1-\frac{4}{n}$, we have
$|V_1|\geq 0.14(1-\epsilon)M_1$, $|V_2|\geq 0.14(1-\epsilon)M_2$.

Finally, we consider the cost of the solution. 
By the rounding algorithm, we have $\Pr[X_1(j)=1] = x'_1(j)$ for all $j\in T$.
By Lemma \ref{lemma:expxprime}, we have
$E[X_1(j)] = E[x'_1(j)] \leq \frac{\eta}{2}x_1(j)$, and
$E[X_2(j)]\leq \frac{\eta}{2}x_2(j)$. Therefore,
\begin{eqnarray*}
E[X_1+X_2] &=& \sum_j E[X_1(j)]+E[X_2(j)] \\
&\leq& \frac{\eta}{2} \sum_j x_1(j)+x_2(j)\leq \eta B/2
\end{eqnarray*}
By Markov's inequality, we have $\Pr[X_1+X_2 > \eta B]\leq \frac{1}{2}$.

Putting everything together, with probability at least $\frac{1}{2} - \frac{2}{n}$,
we have $X_1+X_2 \leq \eta B$ and 
$|V_1|\geq 0.14(1-\epsilon)M_1$, $|V_2|\geq 0.14(1-\epsilon)M_2$.
\end{proof}

\subsection{Approximation for the \proboverlap{} problem}
\label{sec:overlap}

We consider the variation in which the overlap between descriptors is bounded
for the case $k=2$. So the objective is to find $T_1, T_2$ such that
$T_1$ and $T_2$ cover $C_1$ and $C_2$, $|T_1|+|T_2|\leq B$, and $|T_1\cap T_2|\leq B_o$,
where $B_o$ is an input parameter which captures the allowed overlap.
Consider the following LP.

\begin{eqnarray*}
\max \sum_i z(i) && \text{s.t.}\\
\forall s_i\in C_1: \sum_{j\in t_i} x_1(j) &\geq& z(i)\\
\forall s_i\in C_2: \sum_{j\in t_i} x_2(j) &\geq& z(i)\\
\sum_j x_1(j) + x_2(j) &\leq& B\\
\sum_{s_i\in C_1} z(i) \geq M_1 &\;,&
\sum_{s_i\in C_2} z(i) \geq M_2\\
\forall j: x_1(j) + x_2(j) & \leq & 1 + y(j)\\
\sum_j y(j) &\leq& B_o
\end{eqnarray*}
All the variables of $(\mathcal{P})$ are in range $[0,1]$.

Our rounding involves the following steps
\begin{itemize}
\item
Let $x, y, z$ be the optimal fractional solution for the above LP.
\item
For $j=1,\ldots,m$ in turn, and independently:
\begin{itemize}
\item
With probability $y(j)$: $X_1(j)=1, X_2(j)=1$
\item
With probability $x_1(j)-y(j)$: $X_1(j)=1, X_2(j)=0$
\item
With probability $x_2(j)-y(j)$: $X_1(j)=0, X_2(j)=1$
\item
With probability $1-x_1(j)-x_2(j)+y(j)$: $X_1(j)=X_2(j)=0$
\end{itemize}
\end{itemize}

\begin{theorem} \label{theorem:roundoverlap}
Suppose an instance of \proboverlap{} satisfies the following conditions, with $k=2$: (1) $M_{\ell} \geq a |C_{\ell}|$ for $\ell\in[k]$, and for some constant $a\in(0, 1]$, and (2) $(\Delta+1) \leq \min_{\ell}\frac{d |C_{\ell}|}{\log n}$ for a constant $d$.
If the LP relaxation $(\mathcal{P})$ is feasible,
with probability at least $1-\frac{1}{n}$, the above algorithm successfully returns a solution $X$, which ensures $Z_{\ell}\geq M_{\ell}/8$, $|X_1|+|X_2|\leq 3B$, and $|X_1\cap X_2|\leq 3B_o$.
\end{theorem}
\begin{proof}  
Our proof follows on the same lines as the proof of Theorem \ref{theorem:round1}. Within a single round of Step 3 of the algorithm, with high probability we have $Z_1\geq M_1/8$, $Z_2\geq M_2/8$. The rounding ensures $\Pr[X_1(j)=1] = x_1(j)$, and $\Pr[X_2(j)=1] = x_2(j)$. This implies that $E[|X_1|+|X_2|]\leq B$, so that by the Markov inequality, $\Pr[|X_1|+|X_2|>3B]\leq 1/3$. Also, by the rounding, we have $\Pr[X_1(j)=X_2(j)]=y(j)$, for any $j$. Therefore, $E[|X_1\cap X_2|]=\sum_j y(j)\leq B_o$. By the Markov inequality, $\Pr[|X_1\cap X_2|>3B_0]\leq 1/3$. Therefore, each round of Step 3 of the algorithm succeeds with probability at least $1/3-1/n$, which in turn implies the high probability bound.
\end{proof}

\subsection{Dynamic Programming algorithm when $\gamma$ is small}
\label{sse:dp}

We show that a simple dynamic program  gives a $(\frac{1}{\gamma}, 1)$-approximation, when $k$ is a constant. To avoid notational clutter, we describe the algorithm here for $k=2$, and the extension to larger $k$ is direct.

We assume the tags in $T$ are ordered $1,\ldots,m$.
Consider $X_1, X_2\subseteq T$.
Let $V_1(X_1, X_2)=\{s_i\in C_1: t_i \cap X_1\neq\emptyset\}$ and 
$V_2(X_1, X_2)=\{s_i\in C_2: t_i\cap X_2\neq\emptyset\}$.
For a set $T$, let $\mathbf{1}(T\neq\emptyset)$ be an indicator variable which is $1$ if $T\neq\emptyset$.
Using this notation, we can write 
$|V_1(X_1, X_2)| = \sum_{i\in C_1}\mathbf{1}\{t_i\cap X_1\neq\emptyset\}$.
Let $w_1(X_1)=\sum_{j\in X_1} |E_1(j)|$ and 
$w_2(X_2)=\sum_{j\in X_2} |E_2(j)|$, where $E_k(j) = \{s_i \in C_k: j \in t_i\}$

Algorithm \algodp{} involves the following steps:
\begin{itemize}
\item
For integers $j\leq m$, $\ell\leq B$, $w_1\leq M_1$, $w_2\leq M_2$,
we maintain information $S(j, \ell, w_1,  w_2)$ in the following manner:
\begin{itemize}
\item
If there exists a solution $X_1,  X_2\subseteq\{1,\ldots,j\}$ such that
$X_1\cap X_2 =\phi$, $|X_1|+|X_2|\leq\ell$, $w_1(X_1)=w_1$ and $w_2(X_2)=w_2$, then
$S(j, \ell, w_1, w_2)=1$ 
\item
If no such solution $X_1, X_2\subseteq\{1,\ldots,j\}$ exists,
then $S(j, \ell, w_1, w_2)=0$.
\end{itemize}
\item
We compute the entries of $S(j, \ell, w_1, w_2)$ using the following recurrence, which can
be computed in a bottom up manner:
\begin{itemize}
\item
We initialize $S(1, 1, |E_1(1)|, 0)=1$, $S(1, 1, 0, |E_2(1)|)=1$, and $S(1, 0, 0, 0)=1$.
We have $S(1, a, b, c)=0$ for all other choices of $a, b, c$. 
\item
Next, consider $S(j, \ell, w_1, w_2)$ for $j>1$. 
Then, $S(j, \ell, w_1, w_2)=1$ if one of the following holds: $S(j-1,\ell, w_1, w_2)=1$,
or $S(j-1, \ell-1, w_1-|C_1(j)|, w_2)=1$ or $S(j-1, \ell-1, w_1, w_2-|C_2(j)|)=1$.
Otherwise, $S(j, \ell, w_1, w_2)=0$.
\end{itemize}
\end{itemize}

\begin{theorem}
\label{theorem:dp}
If there exists a feasible solution to an instance of \prob{}, then algorithm \algodp{} computes $S(m, B, M_1, M_2)=1$, and,  there exists a solution $(X'_1, X'_2)$ which is a $(\frac{1}{\gamma}, 1)$-approximation.
\end{theorem}
\begin{proof}
First, we observe that the above dynamic programming algorithm  correctly computes 
$S(j, \ell, w_1,  w_2)$ for all inputs by induction.
The base case is $j=1$.

Suppose there is a feasible solution $(X_1, X_2)$ such that
$|X_1|+|X_2|\leq  B$, $|V_1(X_1, X_2)|\geq M_1$, $|V_2(X_1, X_2)|\geq M_2$.
Then,
\begin{eqnarray*}
|V_1(X_1, X_2)| &=& \sum_{s_i\in C_1}\mathbf{1}\{t_i\cap X_1\neq\emptyset\} \\
&\leq& \sum_{s_i\in C_1}|t_i \cap X_1| = \sum_{j\in X_1}|E_1(j)|.
\end{eqnarray*}
This implies $w_1(X_1)=\sum_{j\in X_1}|E_1(j)|\geq M_1$ if $|V_1(X_1, X_2)|\geq M_1$.
Similarly, $w_2(X_2)=\sum_{j\in X_2}|E_2(j)|\geq M_2$ if $|V_2(X_1, X_2)|\geq M_2$.
Therefore, if there is a solution $(X_1, X_2)$ with $|X_1|+|X_2|\leq B$ and
$|V_1(X_1, X_2)|\geq M_1$, $|V_2(X_1, X_2)|\geq M_2$, then 
$S(m, B, M_1, M_2)=1$, and the dynamic program correctly determines this.

Suppose $S(m, B, M_1, M_2)=1$. Then, there exists a solution $(X_1, X_2)$ (which can be
computed by tracing the steps of the dynamic program),  such that
$|X_1|+|X_2|\leq B$ and $w_1(X_1)\geq M_1$, $w_2(X_2)\geq M_2$.
From the above analysis,
\[
\sum_{j\in X_1}|E_1(j)| = \sum_{i\in C_1}|t_i\cap X_1| = \sum_{i\in C_1} |t_i|\mathbf{1}\{t_i \cap X_1\neq\emptyset\}
\]

Therefore, if $\sum_{j\in X_1}|E_1(j)|\geq M_1$, we have
\[
\sum_{i\in C_1} \gamma\mathbf{1}\{t_i\cap X_1\neq\emptyset\} \geq
\sum_{i\in C_1} |t_i|\mathbf{1}\{t_i\cap X_1\neq\emptyset\} \geq M_1,
\]
which implies
\[
|V_1(X_1, X_2)| = \sum_{i\in C_1}\mathbf{1}\{t_i\cap X_1\neq\emptyset\} \geq  \frac{1}{\gamma}M_1. 
\]
Similarly, we have $|V_2(X_1, X_2)|\geq \frac{1}{\gamma}M_2$.
Therefore, $(X_1, X_2)$  is an $(\frac{1}{\gamma}, 1)$-approximate solution.
\end{proof}

\section{Approximation using submodularity}
\label{sec:submod}
The \prob{} problem can be viewed as a problem of submodular function maximization with constraints which can be expressed as a matroid.
For convenience, we assume that a cost budget $B$ is also specified as part of an instance of \prob{} and that the goal of the $(\alpha,\delta)$--approximation algorithm is to produce a solution
that covers  at least $\alpha M_{\ell}$ objects in each cluster $C_{\ell}$ and has a cost of at most $\delta B$.
This assumption can be made without loss of generality 
since the optimal cost is an integer
in  $[1 ~..~|T|]$; one can do a binary search over this
range by executing the algorithm with $O(\log{|T|})$ different budget values and using the smallest budget for which the algorithm produces a 
solution.
We first discuss the necessary concepts, and then describe our algorithm. We refer the reader to \cite{Calinescu:2011:MMS:2340436.2340447} 
for more details regarding submodular function maximization subject to matroid constraints.

A \underline{matroid}  is a pair $\mathcal{M}=(Y,\mathcal{I})$, where $\mathcal{I}\subseteq 2^Y$ and
(1) $\forall A'\in\mathcal{I}$, $A\subset A'\Rightarrow A\in\mathcal{I}$, and
(2) $\forall A, A'\in\mathcal{I}$, $|A|<|A'|\Rightarrow \exists x\in A'-A$ such that $A \cup \{x\} \in\mathcal{I}$.
A function $f:2^Y\rightarrow\mathbb{R}_{\geq 0}$ is 
\underline{submodular} if
$f(A \cup \{x\})-f(A)\geq f(A' \cup \{x\})-f(A')$ for all $A\subseteq A'$.
Function $f(\cdot)$ is \underline{monotone} if $f(A)\leq f(A')$ for all $A\subseteq A'$.

\smallskip
\noindent
\textbf{Constructing a matroid for \prob{}.}
For each tag $j\in T$, let $Y_j=\{a_j, b_j\}$. Let $Y=\cup_j Y_j$.
Let $\mathcal{I}=\{A\subset Y: |A\cap Y_j|\leq 1,\ \forall j 
\mbox{ and }|A|\leq B\}$.
Then $\mathcal{M}=(Y, \mathcal{I})$ can be seen as an intersection of a partition matroid,
which requires $|A\cap Y_j|\leq 1$ for all $j$, and a uniform matroid, which requires $|A|\leq B$.

\begin{lemma}
\label{lemma:matroid}
$\mathcal{M}=(Y, \mathcal{I})$ is a matroid.
\end{lemma}

\noindent
\textbf{Constructing a submodular function.}
It is easy to verify that the function $|V_{\ell}(X)|/M_{\ell}$ (which is the fraction of objects covered by solution $X$ in cluster $C_{\ell}$) is a submodular function of $X$. When $k=2$, we need to find a solution $X$ such that $|V_1(X)|/M_1\geq 1$ and $|V_2(X)|/M_2\geq 1$ hold \emph{simultaneously}. This can be achieved by requiring $\min\Big\{\frac{|V_1(X^A)|}{M_1}, \frac{|V_2(X^A)|}{M_2}\Big\} \geq 1$. However, the minimum of two submodular functions is not submodular, in general. We handle this by using the ``saturation'' technique of \cite{Krause2008RobustSO}:
for $A\subseteq Y$, define $X^A_1=\{j\in T: a_j\in A\}$ and $X^A_2=\{j\in T: b_j\in A\}$,
and let $X^A=(X^A_1, X^A_2)$.
Define $F_1(A)=\min\Big\{\frac{|V_1(X^A)|}{M_1}, 1\Big\}$, $F_2(A)=\min\{\frac{|V_2(X^A)|}{M_2}, 1\}$ and
$F(A) = F_1(A) + F_2(A)$. It is easy to verify the
following lemma.

\begin{lemma}
\label{lemma:submod}
$F_1(A)$, $F_2(A)$, and $F(A)$ are monotone submodular functions of $A$.
\end{lemma}
\noindent
Our algorithm  for \prob{} with $k = 2$ involves the following steps.

\noindent
1. Use the algorithm of \cite{Calinescu:2011:MMS:2340436.2340447} to find a set $A\in\mathcal{I}$ which
maximizes $F(A) = F_1(A) + F_2(A)$.

\noindent
2. Return the solution $X^A=(X^A_1, X^A_2)$.

\smallskip

\begin{theorem}\label{thm:half-one}
Suppose there is a feasible solution to an instance $(S, \pi, T, B, M_1, M_2)$ of \prob{},
with $k=2$. Then, the above algorithm runs in polynomial time and returns an $(1-2/e, 1)$-approximate solution $X^A$.
\end{theorem}

In \cite{sambaturu2019}, we show how the above approach can be
extended to approximate the objective of maximizing the 
total coverage (i.e., $\sum_{\ell}|V_{\ell}(X)|$), for any $k$.

\section{Experimental results}
\label{sec:results}

Our experiments focus on  the following questions.\\
\textbf{1. Benefit of allowing cluster specific coverage.} How do the results from our \prob{} formulation compare with those of \cite{DGR-NIPS-2018}?\\
\textbf{2. Dependence of the cost on coverage.} Does the cost increase gradually as the coverage requirement increases?\\
\textbf{3. Descriptions with pairs of tags}. Does adding pair of tags to the tagset provide better explanations of clusters? How do these results compare to those where pairs of tags are not used?\\
\textbf{4. Performance}. Does \algo{} give solutions with good approximation guarantees in practice, and does it scale to large real world datasets?\\
\textbf{5. Explanation of clusters}. Do the solutions provide interpretable explanations of clusters in real world datasets?

\subsection{Datasets and methods}

\noindent
\textbf{Datasets.}
Table \ref{tab:datasets} provides details of the real and synthetic datasets used in our experiments. In the synthetic datasets, an object is associated with a tag with probability $p$. 

\begin{table}[!ht]
    \centering
    \begin{tabular}{c c c c c} \hline
        \textbf{Real/Synthetic} & $|S|$ & $|T|$ & $|C_1|$ & $|C_2|$\\ \hline  
         Genome (Threat) & 248 & 4632 & 73  & 175 \\
         Uniref90 & 21537 & 2193 & 13406 & 8131 \\
         Flickr & 2455 & 175 & 1052 & 1402 \\
         Philosophers & 249 & 14549 & 110 & 139  \\
         Synthetic-1 &  100 & 100 & 48 & 52\\ 
         Synthetic-2 &  1000 & 1000 & 502    & 498\\
         Synthetic-3 &   1000 & 1000 & 478   & 522 \\ \hline
    \end{tabular}
    \caption{Description of datasets. The three synthetic datasets above were
     generated using probability values 0.05, 0.2 and 0.05 respectively.}
    \label{tab:datasets}
\end{table}
The Threat and Uniref90 datasets \cite{Jain-etal-2018b,RW-2018} contain genome sequences and information that may indicate a given gene's threat potential, which is established manually by domain experts--- this is used to partition the sequences into four clusters (referred to as threat bins 1--4). The tags associated with these sequences are various characteristics of the genes in them, obtained from Bioinformatics repositories.
Uniref90 is an expanded version of the Threat dataset, with additional attributes computed using sequence similarity.

The Flickr dataset \cite{YangML13} consists of images as nodes and relationships between images as edges. A relationship could correspond to images being submitted from the same location, belonging to the same group, or sharing common tags, etc. We use the Louvain algorithm in Networkx \cite{HagbergSS08} to generate communities of images, and pick the communities as clusters. User defined tags, such as ``dog'', ``person'', ``car'', etc., are provided for each image. The Philosophers dataset \cite{YangML13} consists of Wikipedia articles (considered to be the objects to be clustered) on various philosophers. The tags corresponding to each object are the non-philosopher Wikipedia articles to which there is an outlink from the philosopher article. The clusters in the philosopher data are generated by grouping communities that share a common keyword as a single cluster.  The Synthetic-2 and Synthetic-3 datasets are generated with four clusters. In some experiments, we merge the clusters in these datasets into two clusters, one corresponding to clusters 1 and 2, and the other corresponding to clusters 3 and 4, as shown in Table \ref{tab:datasets}. In many of our experiments, we consider $k=2$, and fixed $M_1$ and $M_2$ close to $70\%$ of that of $|C_1|$ and $|C_2|$, respectively. The Twitter dataset used in \cite{DGR-NIPS-2018} was unavailable due to the terms of the dataset, and we are unable to compare with the results of \cite{DGR-NIPS-2018}.

\noindent
\textbf{Methods.}
We study the performance of \algo{} in our experiments. We run the rounding steps 4-14 in Algorithm \ref{alg:round} for 10 iterations.
We use the ILP as a baseline. Note that for the complete coverage version (i.e., $M_{\ell}=|C_{\ell}|$), the ILP is exactly the method used by Davidson et al. (\cite{DGR-NIPS-2018}). 

\noindent
\textbf{Code.} The code is available at \url{https://github.com/prathyush6/ExplainabilityCodeAAAI20.git}.

\subsection{Results}
\textbf{1. Benefit of allowing cluster specific coverage.}
The exact coverage formulation of \cite{DGR-NIPS-2018} (which corresponds to $M_{\ell}=|C_{\ell}|$ for all $\ell$) is infeasible for some of the datasets we consider. Instead, we examine the cluster descriptions computed using the cover-or-forget formulation of \cite{DGR-NIPS-2018}, which maximizes the total number of objects covered. 
Figures \ref{fig:flickr_boc} and \ref{fig:uniref_boc} show the coverage percent for each cluster, i.e., $(|V_{\ell}(X)|/M_{\ell})\times 100\%$, ($y$-axis) versus the cost of the solution ($x$-axis), for the Flickr and Uniref90 datasets, respectively. Both figures show that the coverage is highly imbalanced. For instance, with $3$ tags, almost 90\% of elements in cluster C2 are covered, whereas only 57\% of elements in C1 are covered in Figure \ref{fig:flickr_boc}. This is a limitation of the cover-or-forget approach, and the cluster specific coverage requirements in \prob{} can help alleviate this problem.

\begin{figure}
\centering
\subfloat[Flickr]{%
\label{fig:flickr_boc}%
\includegraphics[scale= 0.5]{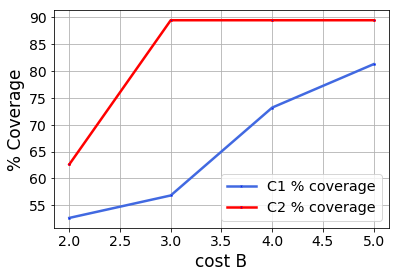}}%
\qquad
\subfloat[Uniref90]{%
\label{fig:uniref_boc}%
\includegraphics[scale =0.5]{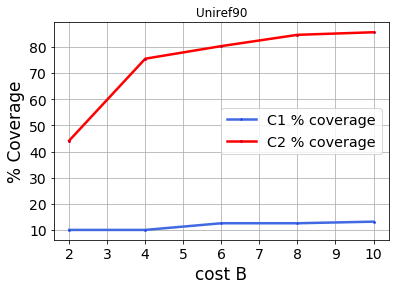}}%
\caption{Coverage percent in each cluster ($y$-axis) and the solution cost ($x$-axis) for the Flickr and Uniref datasets. 
}
\end{figure}

\noindent
\textbf{2. Dependence of cost on the coverage requirement.}
Figures \ref{fig:2clus} and \ref{fig:4clus} 
show the cost of the solution vs the coverage fraction. Initially, the cost grows slowly, but after a point, the cost increases rapidly. For some parameter settings, there is no feasible solution, which corresponds to the ends of the curves. As the number of clusters increases, the cost to cover a given fraction of elements increases.

\begin{figure}
\centering
\subfloat[$k=2$ clusters]{%
\label{fig:2clus}
\includegraphics[scale = 0.5]{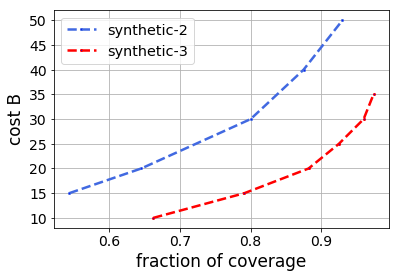}}%
\qquad
\subfloat[$k=4$ clusters]{%
\label{fig:4clus}%
\includegraphics[scale = 0.5]{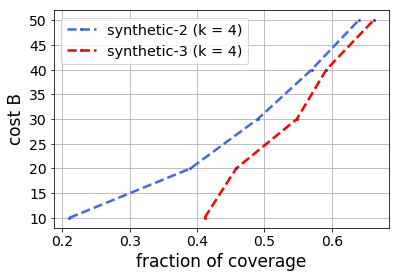}}
\caption{Overall fraction of coverage ($x$-axis) vs the cost $B$ ($y$-axis). The minimum coverage requirement in each cluster is set to at least 50\%.
}
\end{figure}

\noindent
\textbf{3. Descriptions with pairs of tags.}
We extend the set of tags $T$ to $T_{ext}$ by adding every pair $(j, j')$, where $j, j'\in T$, and use $T_{ext}$ for finding descriptions. For some datasets, this increases the feasible regime, but when the instance is feasible, the solutions using $T$ and $T_{ext}$ are pretty close. However, even if the description cost is very similar, using $T_{ext}$ sometimes provides more meaningful descriptions. For instance, on Philosophers dataset, we found pairs such as (`Benedict\_XIV', `Roman\_Catholic\_religious\_order') picked to describe the cluster corresponding to Wikipedia articles related to Christianity.

\noindent
\textbf{4. Performance.} 
First, we consider the approximation guarantee of \algo{} in practice.
Figure \ref{fig:approx} shows the approximation ratios (i.e., the ratio of the coverage achieved by \algo{}, to that of an optimum solution) on the $y$-axis, and the solution cost on the $x$-axis. Recall that the analysis in Theorem \ref{theorem:round1} only guarantees a coverage factor of $1/8 = 0.125$, but the plot for $k=2$ shows that the approximation factors in practice are much higher---they are always $\geq 0.8$, and $>0.9$ in most cases. Note that the curves are non-monotone---this is due to the stochastic nature of \algo{}. However, for $k=4$, the approximation ratios are lower as shown in Figure \ref{fig:kapprox}.

\begin{figure}%
\centering
\subfloat[Datasets with 2 clusters]{%
\label{fig:approx}%
\includegraphics[scale = 0.5]{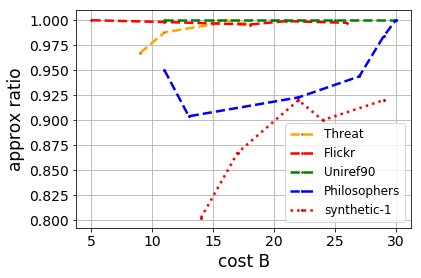}}%
\qquad
\subfloat[Datasets with 4 clusters]{%
\label{fig:kapprox}%
\includegraphics[scale = 0.5]{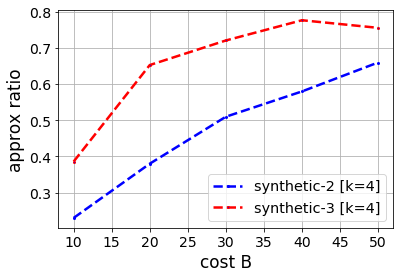}}%
\caption{Approximation ratio of \algo{}  ($y$-axis) vs  budget ($x$-axis) for different datasets (higher is better).
}
\end{figure}

We also observe that \algo{} is quite scalable.
The running time is dominated by the time needed to solve the LP. We use the Gurobi solver, which is able to run successfully on datasets whose data matrix (i.e., the matrix of objects and tags) has up to $10^8$ entries. In contrast, the ILP does not scale beyond datasets with more than $10^6$ entries.

\noindent
\textbf{5. Explanation of clusters}

\noindent
\textbf{(a) Genomic threat sequences (harmful v harmless).}
Our method chose 13 tags for the harmful cluster.
Upon expert review of our results, we found that certain tags served as indicators that genes found within the harmful cluster can intrinsically be viewed as harmful, while others may need to act in concert, be viewed in combination with other tags, or be representative of selection bias. Of the 13 tags selected, 4 indicate intrinsic capability of being harmful: KW-0800 (toxin), 155864.Z3344 (Shiga~toxin~1), 
IPR011050 (Pectin lyase fold/virulence), and 
IPR015217 (invasin domain). Another 4 tags are suggestive that the genes implicated are involved in processes or locations commonly associated with threat: KW-0732 (signal peptide), KW-0614 (plasmid), KW-0964 (secreted), and GO:0050896 (response~to~stimulus). Other tags associated with the threat partition such as KW-0002 (3-D structure) indicate a limited amount of data and perhaps bias in the research literature for the clusters analyzed. Table \ref{tab:genomedata} provides the details.

To define the clusters, each gene was used as a seed to obtain constituent members of Uniref90 groups. By including genes greater than or equal to 90\% sequence identity to the manually curated set, clusters were created and the number of sequences with associated attributes increased to 63,305.

\begin{table*}[!ht]
    \centering
    \begin{tabular}{p{2.4cm} p{2.8cm} p{11cm}}
        \hline
        \textbf{String} & \textbf{Keyword} & \textbf{Definition} \\ \hline 
        \red{KW-0800} & \red{Toxin}    & {Naturally-produced poisonous protein that damages or kills other cells, or the producing cells themselves in some cases in bacteria. Toxins are produced by venomous and poisonous animals, some plants, some fungi, and some pathogenic bacteria. Animal toxins (mostly from snakes, scorpions, spiders, sea anemones and cone snails) are generally secreted in the venom of the animal.}\\ \hline
        \blue{GO:0050896} & \blue{response to \newline stimulus} &    
         {Any process that results in a change in state or activity of a cell or an organism (in terms of movement, secretion, enzyme production, gene expression, etc.) as a result of a stimulus. The process begins with detection of the stimulus and ends with a change in state or activity or the cell or organism.} \\ \hline
        \blue{KW-0964} & \blue{secreted} & {Protein secreted into the cell surroundings.} \\ \hline
        \black{GO:0050794} & \black{regulation \newline of cellular \newline  process} & {Any process that modulates the frequency, rate or extent of a cellular process, any of those that are carried out at the cellular level, but are not necessarily restricted to a single cell. For example, cell communication occurs among more than one cell but occurs at the cellular level.} \\ \hline
        \black{GO:0016787} & \black{hydrolase \newline activity}    & {Catalysis of the hydrolysis of various bonds, e.g. C-O, C-N, C-C, phosphoric anhydride bonds, etc. Hydrolase is the systematic name for any enzyme of EC class 3.}\\ \hline
        \red{IPR011050} &    \red{Pectin\_lyase-fold/\newline virulence}    & {Microbial pectin and pectate lyases are virulence factors that degrade the pectic components of the plant cell wall.}\\ \hline
        \red{IPR015217} & \red{Invasin\_dom\_3} & {It forms part of the extracellular region of the protein, which can be expressed as a soluble protein (Inv497) that binds integrins and promotes subsequent uptake by cells when attached to bacteria.}
        \\ \hline
        \blue{KW-0732}    & \blue{Signal} & {Protein which has a signal sequence, a peptide usually present at the N-terminus of proteins and which is destined to be either secreted or part of membrane components.} \\ \hline
        \black{KW-0002} & \black{3-D  structure\newline (KW-0002)} &    {P+A1:C14 or part of a protein, whose three-dimensional structure has been resolved experimentally (for example by X-ray crystallography or NMR spectroscopy) and whose coordinates are available in the PDB database.}\\ \hline
    \end{tabular}
 \caption{Tags selected by our algorithm for the harmful cluster in the Threat dataset. \red{Red} is intrinsic threat. \blue{Blue} is suggestive. Black is due to a lack of sufficient background.}
    \label{tab:genomedata}
\end{table*}

\smallskip

\noindent
\textbf{(b) Philosophers dataset:}~ Here, Cluster~1 is the set of Wikipedia pages related to India and Greece, whereas Cluster~2 has pages
related to Christianity. 
The tags picked by our algorithm to explain Cluster~1 are
\textit{Buddhist terms and concepts, Metaphysics, Sanyasa, Buddhism, Athenian, Mathematician, Greek Language}, which are consistent with the pages in
the cluster.
The tags picked to explain Cluster~2 are \textit{Constantinople, Existentialism, Abortion, Political Philosophy, Theology, England}, which are consistent
with the contents of that cluster.
\noindent

\section{Conclusions}
\label{sec:conc}

We formulated a version of the cluster description problem that allows simultaneous coverage requirements for all the clusters. We presented rigorous approximation algorithms for the problem using techniques from randomized rounding of linear programs and submodular optimization.  Our rounding-based algorithm exhibits very good performance in practice. Using a real world data set containing genomic threat sequences, we observed that the descriptors found using the algorithm give useful insights. In our experiments, we considered several different parameters including coverage level, budget, and overlap. Using these parameters, our approach can be used to obtain a range of solutions from which a practitioner can choose appropriate descriptors.

\bibliographystyle{plain}
\bibliography{references}
\end{document}